\documentclass{llncs}
\usepackage[utf8x]{inputenc}
\usepackage{amsmath,amssymb}
\usepackage{tikz}
\usepackage{tikz-qtree}
\usetikzlibrary{shapes.geometric,shapes.misc}
\usetikzlibrary{positioning,calc}
\usepackage{algorithmicx}
\usepackage[noend]{algpseudocode}
\usepackage{hyperref}
\usepackage{subcaption}
\usepackage{enumitem}  

\def\NP{\mathbf{NP}}
\def\root{{\textsc{root}}}

\def\rank{{\textsc{rank}}}
\def\parent{{\textsc{parent}}}
\def\merge{{\textsc{merge}}}
\def\compress{\textsc{collapse}}
\def\children{{\textsc{children}}}
\def\push{{\textsc{push}}}

\def\levelone#1{\children(#1)}
\def\Nat{{\mathbb{N}_0}}

\begin{document}
\title{Recognizing Union-Find trees built up using union-by-rank strategy is $\NP$-complete\thanks{Research was supported by the NKFI grant no.~108448.}}
\author{Kitti~Gelle,~Szabolcs~Iv\'an}
\institute{Department of Computer Science, University of Szeged, Hungary\\ \email{\{kgelle,szabivan\}@inf.u-szeged.hu}}
%\address[kitti,szabivan]{University of Szeged, Hungary}
%\tnotetext[thanksfn]{This research was supported by\ldots}
\maketitle
\begin{abstract}
	Disjoint-Set forests, consisting of Union-Find trees, are data structures having a widespread practical application due to their efficiency.
	Despite them being well-known, no exact structural characterization of these trees is known (such a characterization exists for 
	Union trees which are constructed without using path compression) for the case assuming
	union-by-rank strategy for merging.
	In this paper we provide such a characterization by means of a simple $\push$ operation
	and show that the decision problem whether a given tree (along with the rank info of its nodes)
	is a Union-Find tree is $\NP$-complete,
	complementing our earlier similar result for the union-by-size strategy.
\end{abstract}
\section{Introduction}
Disjoint-Set forests, introduced in~\cite{Galler:1964:IEA:364099.364331},
are fundamental data structures in many practical algorithms where one has to maintain a partition of some set,
which supports three operations: \emph{creating} a partition consisting of singletons, \emph{querying} whether two given elements are
in the same class of the partition (or equivalently: \emph{finding} a representative of a class, given an element of it)
and \emph{merging} two classes.
Practical examples include e.g.
building a minimum-cost spanning tree of a weighted graph~\cite{Cormen:2001:IA:580470},
unification algorithms~\cite{Knight:1989:UMS:62029.62030}
etc.

To support these operations, even a linked list representation suffices but to achieve an almost-constant amortized time cost per operation, Disjoint-Set forests are used in practice. In this data structure, sets are represented as directed trees with the edges directed towards the root; the $\textsc{create}$ operation creates $n$ trees having one node each (here $n$ stands for the number of the elements in the universe), the
$\textsc{find}$ operation takes a node and returns the root of the tree in which the
node is present (thus the $\textsc{same-class}(x,y)$ operation is implemented as
$\textsc{find}(x)==\textsc{find}(y)$), and the $\textsc{merge}(x,y)$ operation
is implemented by merging the trees containing $x$ and $y$, i.e. making one of the root nodes to be a child of the other root node (if the two nodes are in different classes).

In order to achieve near-constant efficiency, one has to keep the (average) height of the trees small.
There are two ``orthogonal'' methods to do that: first, during the merge operation it is advisable to attach the ``smaller'' tree below the ``larger''
one. If the ``size'' of a tree is the number of its nodes, we say the trees are built up according to the \emph{union-by-size} strategy,
if it's the depth of a tree, then we talk about the \emph{union-by-rank} strategy. Second, during a $\textsc{find}$ operation invoked on some node $x$
of a tree, one can apply the \emph{path compression} method, which reattaches each ancestor of $x$ directly to the root of the tree in which they are
present. If one applies both the path compression method and either one of the union-by-size or union-by-rank strategies, then any sequence of
$m$ operations on a universe of $n$ elements has worst-case time cost $O(m\alpha(n))$ where $\alpha$ is the inverse of the extremely fast growing
(not primitive recursive) Ackermann function for which $\alpha(n)\leq 5$ for each practical value of $n$ (say, below $2^{65535}$), hence it has
an amortized almost-constant time cost~\cite{Tarjan:1975:EGB:321879.321884}. Since it's proven~\cite{Fredman:1989:CPC:73007.73040} that \emph{any} data structure maintaining a partition has worst-case time cost $\Omega(m\alpha(n))$, the Disjoint-Set forests equipped with a strategy and path compression offer a theoretically optimal data structure which performs exceptionally well also in practice.
For more details see standard textbooks on data structures, e.g.~\cite{Cormen:2001:IA:580470}.

Due to these facts, it is certainly interesting both from the theoretical as well as the practical point of view to characterize those trees
that can arise from a forest of singletons after a number of merge and find operations, which we call Union-Find trees in this paper.
One could e.g. test Disjoint-Set implementations since if at any given point of execution a tree of a Disjoint-Set forest is not a valid
Union-Find tree, then it is certain that there is a bug in the implementation of the data structure (though we note at this point that this
data structure is sometimes regarded as one of the ``primitive'' data structures, in the sense that it is possible to implement a correct
version of them that needs not be certifying~\cite{DBLP:journals/csr/McConnellMNS11}). Nevertheless, only the characterization
of Union trees is known up till now~\cite{DBLP:journals/ipl/Cai93}, i.e. which correspond to the case when one uses one of the union-by- strategies but
\emph{not} path compression. Since in that case the data structure offers only a theoretic bound of $\Theta(\log n)$ on the amortized time cost,
in practice all implementations imbue path compression as well, so for a characterization to be really useful, it has to cover this case
as well.

In this paper we show that the recognition problem of Union-Find trees is $\NP$-complete when the union-by-rank
strategy is used, complementing our earlier results~\cite{DBLP:journals/corr/GelleI15} where
we proved $\NP$-completeness for the union-by-size strategy.
The proof method applied here resembles to that one, but the low-level details for the reduction
(here we use the $\textsc{Partition}$ problem,
there we used the more restricted version $3-\textsc{Partition}$ as this is a very canonical
strongly $\NP$-complete problem) differ greatly.
This result also confirms the statement
from~\cite{DBLP:journals/ipl/Cai93} that the problem ``seems to be much harder''
than recognizing Union trees. As (up to our knowledge) in most of the actual software libraries
having this data structure implemented the union-by-rank strategy is used (apart from the cases
when one quickly has to query the size of the sets as well), for software testing purposes the
current result is more relevant than the one applying union-by-size strategy.

{\bf Related work.} There is an increasing interest in determining the complexity of the recognition problem of various
data structures. The problem was considered for suffix trees~\cite{I2014316,Starikovskaya201514},
(parametrized) border arrays~\cite{I20116959,Lu2002,Duval:2005:BAB:1131983.1131987,MR2544434,MR2894365},
suffix arrays~\cite{Bannai2003208,Duval2002249,Kucherov2013915},
KMP tables~\cite{Duval2009281,Gawrychowski2014337},
prefix tables~\cite{DBLP:conf/stacs/ClementCR09}, 
cover arrays~\cite{Crochemore2010251}, and directed acyclic word- and subsequence graphs~\cite{Bannai2003208}.

\section{Notation}
A \emph{(ranked) tree} is a tuple $t=(V_t,\root_t,\rank_t,\parent_t)$
with $V_t$ being the finite set of its \emph{nodes},
$\root_t\in V_t$ its \emph{root},
$\rank_t:V_t\to\Nat$ mapping a nonnegative integer to each node,
and $\parent_t:(V_t-\{\root_t\})\to V_t$ mapping each nonroot node to its parent
(so that the graph of $\parent_t$ is a directed acyclic graph, with edges being directed towards the root).
We require $\rank_t(x)<\rank_t(\parent_t(x))$ for each nonroot node $x$, i.e. the rank strictly decreases towards the leaves.

For a tree $t$ and a node $x\in V_t$, let $\children(t,x)$ stand for the set $\{y\in V_t:\parent_t(y)=x\}$ of its children
and $\levelone{t}$ stand as a shorthand for $\children(t,\root_t)$, the set of depth-one nodes of $t$.
Also, let $x\preceq_t y$ denote that $x$ is a (non-strict) \emph{ancestor} of $y$ in $t$, i.e. $x=\parent_t^k(y)$ for some $k\geq 0$.
For $x\in V_t$, let $t|_x$ stand for the \emph{subtree} $(V_x=\{y\in V:x\preceq_t y\},x,\rank_t|_{V_x},\parent_t|_{V_x})$ of $t$ rooted at $x$.
As shorthand, let $\rank(t)$ stand for $\rank_t(\root_t)$, the rank of the root of $t$.

Two operations on trees are that of \emph{merging} and \emph{collapsing}.
Given two trees $t=(V_t,\root_t,\rank_t,\parent_t)$ and $s=(V_s,\root_s,\rank_s,\parent_s)$
with $V_t$ and $V_s$ being disjoint and $\rank(t)\geq\rank(s)$, then
their \emph{merge} $\textsc{merge}(t,s)$ (in this order) is the tree
$(V_t\cup V_s,\root_t,\rank,\parent)$ with
$\parent(x)=\parent_t(x)$ for $x\in V_t$, $\parent(\root_s)=\root_t$ and $\parent(y)=\parent_s(y)$ for each nonroot node $y\in V_s$ of $s$,
and
\[\rank(\root_t)=\begin{cases}\rank(t)&\textrm{if }\rank(s)<\rank(t),\\\rank(t)+1&\textrm{otherwise,}\end{cases}\]
and $\rank(x)=\rank_t(x)$, $\rank_s(x)$ resp. for each $x\in V_t-\{\root_r\}$, $x\in V_s$ resp.

Given a tree $t=(V,\root,\rank,\parent)$ and a node $x\in V$, then $\compress(t,x)$ is the tree $(V,\root,\rank,\parent')$
with $\parent'(y)=\root$ if $y$ is a nonroot ancestor of $x$ in $t$ and $\parent'(y)=\parent(y)$ otherwise.
For examples, see Figure~\ref{fig-merge-collapse-push}.

Observe that both operations indeed construct a ranked tree (e.g. the rank remains strictly decreasing towards the leaves).

\begin{figure*}[h]
	\begin{subfigure}[c]{0.5\textwidth}
		\centering
		\begin{tikzpicture}[thick]
		%\draw[white] (0,1.2) rectangle (1,-2.2);
		\node at (1,1) {$s$:};
		\node[circle] (r) at (2.5,1) [draw,inner sep=1.5mm] {};
		\node[right of=r, node distance=4mm] {$r$};
		\node[left of=r, node distance=0mm] {$2$};
		\foreach \i in {1,...,4}{
			\node[circle] (c\i) at (\i,0) [draw,inner sep=1.5mm] {};
			\path (r) edge (c\i);
			\node[left of=c\i, node distance=0mm]{$1$};
		}
		\foreach \i in {0,...,3}{
			\node[circle] (cc\i) at (\i/2+1.25,-1) [draw,inner sep=1.5mm] {};
			\node[left of=cc\i, node distance=0mm]{$0$};
		}
		\path (cc0) edge (c2);
		\path (cc1) edge (c2);
		\path (cc2) edge (c2);
		\path (cc3) edge (c2);
		\node[right of=c2, node distance=4mm] {$x$};
		\node[below of=cc2, node distance=4mm] {$z$};
		\node[circle] (r2) at (6,1) [draw,inner sep=1.5mm] {};
		\node[left of=r2, node distance=10mm]  {$t$:};
		
		\foreach \i in {0,...,1}{
			\node[circle] (c\i) at (\i+5.5,0) [draw,inner sep=1.5mm] {};
			\path (r2) edge (c\i);
			\node[left of=c\i, node distance=0mm]{$0$};
		}
		\node[left of=r2, node distance=0mm] {$2$};
		\node[right of=r2, node distance=5mm] {$y$};
		\end{tikzpicture}
		\caption{Trees $s$ and $t$}
	\end{subfigure}
	\begin{subfigure}[c]{0.5\textwidth}
		\centering
		\begin{tikzpicture}[thick]
		%\draw[white] (0,1.2) rectangle (1,-2.2);
		%\node at (1,1) {$\merge(s,t)$:};
		\node[circle] (r) at (2.5,1) [draw,inner sep=1.5mm] {};
		\node[right of=r, node distance=4mm] {$r$};
		\node[left of=r, node distance=0mm] {$3$};
		\foreach \i in {1,...,4}{
			\node[circle] (c\i) at (\i,0) [draw,inner sep=1.5mm] {};
			\path (r) edge (c\i);
			\node[left of=c\i, node distance=0mm]{$1$};
		}
		\foreach \i in {0,...,3}{
			\node[circle] (cc\i) at (\i/2+1.25,-1) [draw,inner sep=1.5mm] {};
			\node[left of=cc\i, node distance=0mm]{$0$};
		}
		\path (cc0) edge (c2);
		\path (cc1) edge (c2);
		\path (cc2) edge (c2);
		\path (cc3) edge (c2);
		\node[right of=c2, node distance=4mm] {$x$};
		\node[below of=cc2, node distance=4mm] {$z$};
		\node[circle] (r2) at (5,0) [draw,inner sep=1.5mm] {};
		%\node[left of=r2, node distance=10mm]  {$t$:};
		\path (r) edge (r2);
		
		\foreach \i in {0,...,1}{
			\node[circle] (c\i) at (\i+4.5,-1) [draw,inner sep=1.5mm] {};
			\path (r2) edge (c\i);
			\node[left of=c\i, node distance=0mm]{$0$};
		}
		\node[left of=r2, node distance=0mm] {$2$};
		\node[right of=r2, node distance=5mm] {$y$};
		\end{tikzpicture}
		\caption{$t'=\merge(s,t)$}
	\end{subfigure}

	\begin{subfigure}[c]{1\textwidth}
	\centering
	\begin{tikzpicture}[thick]
	%\draw[white] (0,1.2) rectangle (1,-2.2);
	%\node at (1,1) {$\merge(s,t)$:};
	\node[circle] (r) at (2.5,1) [draw,inner sep=1.5mm] {};
	\node[right of=r, node distance=4mm] {$r$};
	\node[left of=r, node distance=0mm] {$3$};
	\foreach \i in {1,3,4}{
		\node[circle] (c\i) at (\i,0) [draw,inner sep=1.5mm] {};
		\path (r) edge (c\i);
		\node[left of=c\i, node distance=0mm]{$1$};
	}
	\node[circle] (c2) at (6.5,-1) [draw,inner sep=1.5mm] {};
	\node[left of=c2, node distance=0mm]{$1$};
	\foreach \i in {0,...,3}{
		\node[circle] (cc\i) at (\i/2+5.75,-2) [draw,inner sep=1.5mm] {};
		\node[left of=cc\i, node distance=0mm]{$0$};
		\path (cc\i) edge (c2);
	}
	\node[right of=c2, node distance=4mm] {$x$};
	\node[below of=cc2, node distance=4mm] {$z$};
	\node[circle] (r2) at (5,0) [draw,inner sep=1.5mm] {};
	%\node[left of=r2, node distance=10mm]  {$t$:};
	\path (r) edge (r2);
	\path (r2) edge (c2);
	\foreach \i in {0,...,1}{
		\node[circle] (ccc\i) at (\i+4.5,-1) [draw,inner sep=1.5mm] {};
		\path (r2) edge (ccc\i);
		\node[left of=ccc\i, node distance=0mm]{$0$};
	}
	\node[left of=r2, node distance=0mm] {$2$};
	\node[right of=r2, node distance=5mm] {$y$};
	\end{tikzpicture}
	\caption{$t''=\push(t',x,y)$}
\end{subfigure}

	\begin{subfigure}[c]{1\textwidth}
	\centering
	\begin{tikzpicture}[thick]
	%\draw[white] (0,1.2) rectangle (1,-2.2);
	%\node at (1,1) {$\merge(s,t)$:};
	\node[circle] (r) at (2.5,1) [draw,inner sep=1.5mm] {};
	\node[left of=r, node distance=4mm] {$r$};
	\node[left of=r, node distance=0mm] {$3$};
	\foreach \i in {1,3,4}{
		\node[circle] (c\i) at (\i,0) [draw,inner sep=1.5mm] {};
		\path (r) edge (c\i);
		\node[left of=c\i, node distance=0mm]{$1$};
	}
	\node[circle] (c2) at (7,0) [draw,inner sep=1.5mm] {};
	\node[left of=c2, node distance=0mm]{$1$};
	\foreach \i in {0,1,3}{
		\node[circle] (cc\i) at (\i/2+6.25,-1) [draw,inner sep=1.5mm] {};
		\node[left of=cc\i, node distance=0mm]{$0$};
		\path (cc\i) edge (c2);
	}
	\node[circle] (cc2) at (8.25,0) [draw,inner sep=1.5mm] {};
	\node[left of=cc2, node distance=0mm]{$0$};
	\path (cc2) edge (r);
	\node[right of=c2, node distance=4mm] {$x$};
	\node[right of=cc2, node distance=4mm] {$z$};
	\node[circle] (r2) at (5,0) [draw,inner sep=1.5mm] {};
	%\node[left of=r2, node distance=10mm]  {$t$:};
	\path (r) edge (r2);
	\path (r) edge (c2);
	\foreach \i in {0,...,1}{
		\node[circle] (ccc\i) at (\i+4.5,-1) [draw,inner sep=1.5mm] {};
		\path (r2) edge (ccc\i);
		\node[left of=ccc\i, node distance=0mm]{$0$};
	}
	\node[left of=r2, node distance=0mm] {$2$};
	\node[right of=r2, node distance=5mm] {$y$};
	\end{tikzpicture}	
		\caption{$t'''=\compress(t'',z)$}
	\end{subfigure}
	\caption{Merge, collapse and push.}
	\label{fig-merge-collapse-push}
\end{figure*}
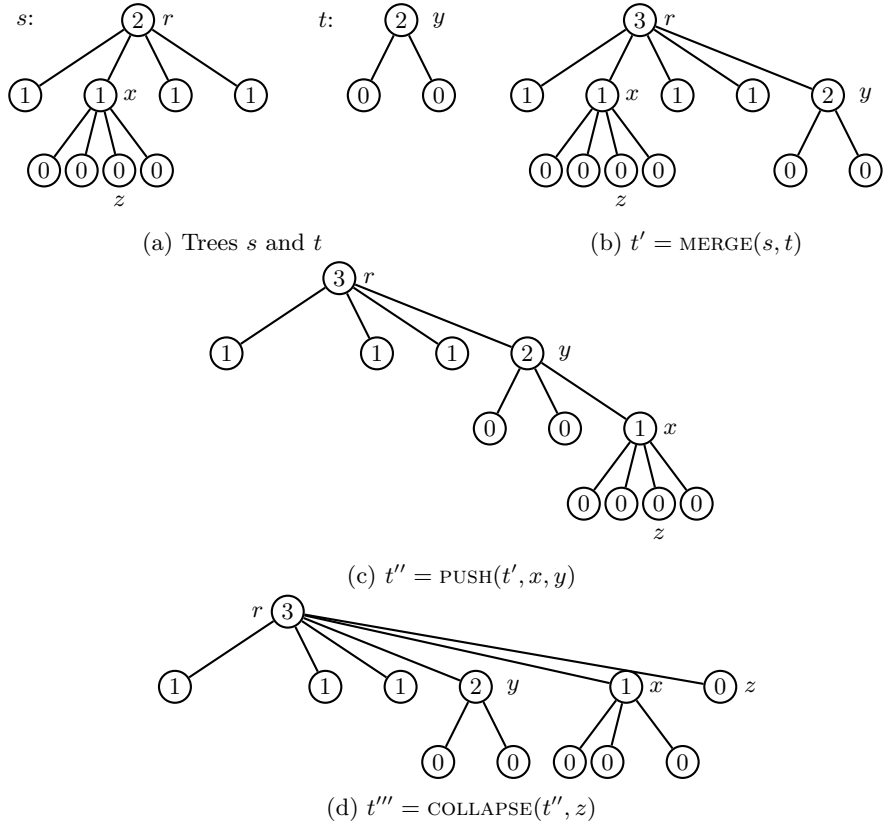

We say that a tree is a \emph{singleton} tree if it has exactly one node, and this node has rank $0$.

The class of Union trees is the least class of trees satisfying the following two conditions:
every singleton tree is a Union tree, and if $t$ and $s$ are Union trees with $\rank(t)\geq\rank(s)$,
then $\textsc{merge}(t,s)$ is a Union tree as well.

Analogously, the class of Union-Find trees is the least class of trees satisfying the following three conditions:
every singleton tree is a Union-Find tree, if $t$ and $s$ are Union-Find trees with $\rank(t)\geq\rank(s)$, then
$\textsc{merge}(t,s)$ is a Union-Find tree as well, and if $t$ is a Union-Find tree and $x\in V_t$ is a node of $t$,
then $\compress(t,x)$ is also a Union-Find tree.

We say that a node $x$ of a tree $t$ \emph{satisfies the Union condition} if
\[\{\rank_t(y):y\in\children(t,x)\}=\{0,1,\ldots,\rank_t(x)-1\}.\]

Then, the characterization of Union trees from~\cite{DBLP:journals/ipl/Cai93} can be formulated in our terms as follows:
\begin{theorem}
\label{thm-union}
A tree $t$ is a Union tree if and only if each node of $t$ satisfies the Union condition.
\end{theorem}

Note that the rank of a Union tree always coincides by its \emph{height}. (And, any subtree of a Union tree is also a Union tree.)
In particular, the leaves are exactly those nodes of rank $0$.

\section{Structural characterization of Union-Find trees}

Suppose $s$ and $t$ are trees on the same set $V$ of nodes, with the same root $\root$ and the same rank function $\rank$.
We write $s\preceq t$ if $x\preceq_s y$ implies $x\preceq_t y$ for each $x,y\in V$.

Clearly, $\preceq$ is a partial order on any set of trees (i.e. is a reflexive, transitive and antisymmetric relation).
It is also clear that $s\preceq t$ if and only if $\parent_s(x)\preceq_t x$ holds for each $x\in V-\{\root\}$
which is further equivalent to requiring $\parent_s(x)\preceq_t \parent_t(x)$ since $\parent_s(x)$ cannot be $x$.

Another notion we define is the (partial) operation $\push$ on trees as follows: when $t$ is a tree and $x\neq y\in V_t$ are \emph{siblings} in $t$,
i.e. have the same parent, and $\rank_t(x)<\rank_t(y)$, then $\push(t,x,y)$ is defined as the tree $(V_t,\root_t,\rank_t,\parent')$ with
\[\parent'(z)=\begin{cases}y&\hbox{if }z=x,\\\parent_t(z)&\hbox{otherwise,}\end{cases}\]
that is, we ``push'' the node $x$
one level deeper in the tree just below its former sibling $y$.
(See Figure~\ref{fig-merge-collapse-push}.)  

We write $t\vdash t'$ when $t'=\push(t,x,y)$ for some $x$ and $y$, and as usual,
$\vdash^*$ denotes the reflexive-transitive closure of $\vdash$.

\begin{proposition}
\label{prop-push-is-lift}
For any pair $s$ and $t$ of trees, the following conditions are equivalent:
\begin{enumerate}[label=(\roman*)]
\item $s\preceq t$,
\item there exists a sequence $t_0=s,t_1,t_2,\ldots,t_n$ of trees such that for each $i=1,\ldots,n$ we have
  $t_i=\push(t_{i-1},x,y)$ for some depth-one node $x\in\levelone{t_{i-1}}$, moreover,
  $\levelone{t_n}=\levelone{t}$ and $t_n|_x\preceq t|_x$ for each $x\in\levelone{t}$,
\item $s\vdash^* t$.
\end{enumerate}
\end{proposition}
\begin{proof}
{\bf{i)$\Rightarrow$ii).}} 
It is clear that $\preceq$ is equality on singleton trees, thus $\preceq$ implies $\vdash^*$ for trees of rank $0$.
Assume $s\preceq t$ for the trees $s=(V,\root,\rank,\parent)$ and $t=(V,\root,\rank,\parent')$ and
let $X$ stand for the set $\children(s)$ of the depth-one nodes of $s$
and $Y$ stand for $\children(t)$.
Clearly, $Y\subseteq X$ since by $s\preceq t$, any node $x$ of $s$ having depth at least two has to satisfy $\parent(x)\preceq_t\parent'(x)$
and since $\parent(x)\neq\root$ for such nodes, $x$ has to have depth at least two in $t$ as well. Now there are two cases:
either $\root=\parent(x)=\parent'(x)$ for each $x\in X$, or $\parent(x)\prec_t\parent'(x)$ for some $x\in X$.

If $\parent'(x)=\root$ for each $x\in X$, then $X=Y$ and we only have to show that $s|_x\preceq t|_x$ for each $x\in X$.
For this, let $u,v\in V(s|_x)$ with $u\preceq_{s|_x}v$. Since $s|_x$ is a subtree of $s$, this holds if and only if
$x\preceq_s u\preceq_s v$. From $s\preceq t$ this implies $x\preceq_t u\preceq_t v$, that is, $u\preceq_{t|_x}v$, hence $s|_x\preceq t|_x$.

Now assume $\parent(x)\prec_t\parent'(x)$ for some $x\in X$.
Then $\parent'(x)\neq\root$, thus there exists some $y\in Y$ with $y\preceq_t\parent'(x)$. By $Y\subseteq X$, this $y$ is a member of $X$ as well,
and $\rank_s(y)=\rank_t(y)>\rank_t(x)=\rank_s(x)$, thus $s'=\push(s,x,y)$ is well-defined.
Moreover, $s'\preceq t$ since $\parent_{s'}(z)\preceq_t z$ for each $z\in V$: either $z\neq x$ in which case $\parent_{s'}(z)=\parent(z)\preceq_t z$
by $s\preceq t$, or $z=x$ and then $\parent_{s'}(z)=y\preceq_t\parent'(x)\preceq_t x=z$ also holds. Thus, there exists a tree $s'=\push(s,x,y)$
for some $x\in\levelone{s}$ with $s'\preceq t$; since $\levelone{s'}=X-\{x\}$, by repeating this construction we eventually arrive to
a tree $t_n$ with $|\levelone{t_n}|=|Y|$, implying $\levelone{t_n}=Y$ by $Y\subseteq \levelone{t_n}$.

{\bf{ii)$\Rightarrow$iii).}} We apply induction on $\rank(s)=\rank(t)$. When $\rank(s)=0$, then $s$ is a singleton tree and the condition in ii) ensures
that $t$ is a singleton tree as well. Thus, $s=t$ and clearly $s\vdash^*t$.

Now let assume the claim holds for each pair of trees of rank less than $\rank(s)$ and let $t_0,\ldots,t_n$ be trees satisfying the condition.
Then, by construction, $s\vdash^*t_n$. Since $\rank(t_n|_x)<\rank(t_n)=\rank(s)$ for each node $x\in\levelone{t_n}$, by $t_n|_x\preceq t|_x$
we get applying the induction hypothesis that $t_n|_x\vdash^*t|_x$ for each depth-one node $x$ of $t_n$, thus $t_n\vdash^* t$, hence
$s\vdash^*t$ as well.

{\bf{iii)$\Rightarrow$ i).}} For $\vdash^*$ implying $\preceq$ it suffices to show that $\vdash$ implies $\preceq$ since the latter is reflexive and transitive.
So let $s=(V,r,\rank,\parent)$ and $x\neq y\in V$ be siblings in $s$  with the common parent $z$, $\rank(x)<\rank(y)$ and let $t=\push(s,x,y)$.
Then, since $\parent_s(x)=z=\parent_t(y)=\parent_t(\parent_t(x))$, we get $\parent_s(x)\preceq_t x$, and by
$\parent_s(w)=\parent_t(w)$ for each node $w\neq x$, we have $s\preceq t$.

\qed\end{proof}

The relations $\preceq$ and $\vdash^*$ are introduced due to their intimate relation to Union-Find and Union trees
(similarly to the case of the union-by-size strategy~\cite{DBLP:journals/corr/GelleI15}, but there the 
$\push$ operation itself was slightly different):
\begin{theorem}
	\label{thm-uf-push}
	A tree $t$ is a Union-Find tree if and only if $t\vdash^* s$ for some Union tree $s$.
\end{theorem}
\begin{proof}
	
	Let $t$ be a Union-Find tree. We show the claim by structural induction.
	For singleton trees the claim holds since any singleton tree is a Union tree as well. 
	Suppose $t=\merge(t_1,t_2)$. Then by the induction hypothesis, $t_1\vdash^* s_1$ and
	$t_2\vdash^* s_2$ for the Union trees $s_1$ and $s_2$. Then, for the tree $s=\merge(s_1,s_2)$ we get that $t\vdash^* s$.
	Finally, assume $t=\compress(t',x)$ for some node $x$. Let $x=x_1\succ x_2\succ\ldots\succ x_k=\root_{t'}$ be the ancestral sequence of $x$ in $t'$. Then, defining $t_0=t$, $t_i=\push(t_{i-1},x_i,x_{i+1})$ we get that $t\vdash^* t_{k-2}=t'$ and $t'\vdash^* s$ for some Union tree $s$ applying the induction hypothesis, thus $t\vdash^* s$ also holds.
	
	Now assume $t\vdash^* s$ (equivalently, $t\preceq s$) for some Union tree $s$.
	We show the claim by induction on the height of $t$.
	For singleton trees the claim holds since any singleton tree is a Union-Find tree.
	
	Now assume $t=(V,\root,\rank,\parent)$ is a tree and $t\vdash^* s$ for some Union tree $s$.
	Then by Proposition~\ref{prop-push-is-lift}, there is a set $X=\children(s)\subseteq\children(t)$
	of depth-one nodes of $t$
	and a function $f:Y\to X$ with $Y=\{y_1,\ldots,y_\ell\}=\children(t)-X$  such that for the sequence $t_0=t$,
	$t_i=\push(t_{i-1},y_i,f(y_i))$ we have that $t_\ell|_x\preceq s|_x$ for each $x\in X$.
	As each $s|_x$ is a Union tree (since so is $s$), we have by the induction hypothesis that
	each $t_\ell|_x$ is a Union-Find tree. Now let $X=\{x_1,\ldots,x_k\}$ be ordered nondecreasingly by rank;
	then, as $s$ is a Union tree and $X=\children(s)$, we get that $\{\rank(x_i)\}=\{0,1,\ldots,\rank(\root)-1\}$
	by Theorem~\ref{thm-union}.
	Hence for the sequence $t'_i$ defined as
	$t'_0$ being a singleton tree with root $\root$ and for each $i\in\{1,\ldots,k\}$,
	$t'_i=\merge(t'_{i-1},t_\ell|_{x_i})$, we get that $t_\ell=t'_k$ is a Union-Find tree.
	Finally, we get $t$ from $t_\ell$ by applying successively one $\compress$ operaton on each node in $Y$,
	thus $t$ is a Union-Find tree as well.
	\qed\end{proof}

\section{Complexity}

In order to show $\NP$-completeness of the recognition problem, we first make a useful observation.

\begin{proposition}
\label{prop-charge}
In any Union-Find tree $t$ there are at least as many rank-$0$ nodes as nodes of positive rank.
\end{proposition}
\begin{proof}
We apply induction on the structure of $t$.
The claim holds for singleton trees (having one single node of rank $0$).
Let $t=\merge(t_1,t_2)$ and suppose the claim holds for $t_1$ and $t_2$. There are two cases.
\begin{itemize}
	\item Assume $\rank(t_1)=0$. Then, since $\rank(t_1)\geq \rank(t_2)$ we have that $\rank(t_2)$ is $0$ as well, i.e. both $t_1$ and $t_2$ are
	  singleton trees (of rank $0$). In this case $t$ has one node of rank $1$ and one node of rank $0$.
	\item If $\rank(t_1)>0$, then (since $\root_{t_1}$ is the only node in $V_t=V_{t_1}\cup V_{t_2}$ whose rank can change at all, in which case it increases)
 	  neither the total number of rank-$0$ nodes nor the total number of nodes with positive rank changes, thus the claim holds.
\end{itemize}
Let $t=\compress(s,x)$ and assume the claim holds for $s$. Then, since the $\compress$ operation does not change the rank of any of the nodes,
the claim holds for $t$ as well.
\qed\end{proof}

In order to define a reduction from the strongly $\NP$-complete problem $\textsc{Partition}$ we introduce several notions on trees:

An \emph{apple} of weight $a$ for an integer $a>0$ is a tree consisting of a root node of rank $2$, a depth-one node of rank $0$ and $a$ depth-one nodes of rank $1$.

A \emph{basket} of size $H$ for an integer $H>0$ is a tree consisting of $H+4$ nodes: the root node having rank $3$,
$H+1$ depth-one children of rank $0$ and one depth-one child of rank $1$, which in turn has a child of rank $0$.

A \emph{flat tree} is a tree $t$ of the following form: the root of $t$ has rank $4$. The immediate subtrees of $t$ are:
\begin{itemize}
	\item a node of rank $0$, having no children;
	\item a node of rank $1$, having a single child of rank $0$;
	\item a node of rank $2$, having two children: a single node of rank $0$ and a node of rank $1$, having a single child of rank $0$;
	\item an arbitrary number of apples,
	\item and an arbitrary number of baskets for some fixed size $H$.
\end{itemize}
(See Figure~\ref{fig-ranktree}.)

\begin{figure}[h!]
\centering\begin{tikzpicture}[every node/.style={draw,thick,circle,inner sep=1pt}, sibling distance=5pt]
\Tree [.$4$ 
	[.$0$ ]
	[.$1$ [.$0$ ]]
	[.$2$ [.$1$ [.$0$ ] ] [.$0$ ]]
	[.$2$ [.$0$ ] [.$1$ ] [.$1$ ] [.$1$ ]]	%alma
	[.$2$ [.$0$ ] [.$1$ ] [.$1$ ] [.$1$ ] [.$1$ ] [.$1$ ]]	%alma
	[.$3$ [.$0$ ] [.$1$ [.$0$ ]] [.$0$ ] [.$0$ ] [.$0$ ] [.$0$ ] [.$0$ ] ] %kosár	
]
%dobozoljunk. (0,0) lesz a fa gyökércsúcsa.. asszem. és jobbra-felfele van a pozitív
%\draw[dotted] (-2.3,-2.5) rectangle (-0.7,-0.5) ; %két átellenes sarok (x1,y1) rectangle (x2,y2) formában
\node (doboz) [rectangle,draw,dotted, minimum width=32mm, minimum height=20mm] at (1.1,-1.5) {};	
\node[draw=none, rectangle, anchor=north] at (doboz.south) {This is an apple};
\node[draw=none, rectangle, node distance=1.6cm, below of=doboz] {of weight $5$};	
\node (doboz) [rectangle,draw,dotted, minimum width=38mm, minimum height=30mm] at (4.8,-2) {};	
\node[draw=none, rectangle, anchor=north] at (doboz.south) {This is a basket};
\node[draw=none, rectangle, node distance=2.1cm, below of=doboz] {of size $5$};	
\node (doboz) [rectangle,draw,dotted, minimum width=22mm, minimum height=20mm] at (-1.7,-1.5) {};	
\node[draw=none, rectangle, anchor=north] at (doboz.south) {This is an apple};
\node[draw=none, rectangle, node distance=1.6cm, below of=doboz] {of weight $3$};	
\node (doboz) [rectangle,draw,dotted, minimum width=22mm, minimum height=30mm] at (-4.0,-2) {};	
\node[draw=none, rectangle, anchor=north] at (doboz.south) {This part is constant};
\node[draw=none, rectangle, node distance=2.1cm, below of=doboz] {in a flat tree};

\end{tikzpicture}
\caption{A flat tree.}
\label{fig-ranktree}
\end{figure}
At this point we recall that the following problem $\textsc{Partition}$ is $\NP$-complete in the strong sense~\cite{Garey:1979:CIG:578533}:
given a list $a_1,\ldots,a_{m}$ of positive integers and a value $k>0$
such that the value $B=\frac{\sum_{i=1}^{m}a_i}{k}$ is an integer,
does there exist a partition $\mathcal{B}=\{B_1,\ldots,B_k\}$ of the set $\{1,\ldots,m\}$ satisfying
$\sum_{i\in B_j}a_i=B$ for each $1\leq j\leq k$?

(Here ``in the strong sense'' means that the problem remains $\NP$-complete even if the numbers are encoded in unary.)

\begin{proposition}
\label{prop-nphard}
	Assume $t$ is a flat tree having $k$ basket children, each having the size $H$,
	and $m$ apple children of weights $a_1,\ldots,a_m$ respectively,
	satisfying $H\cdot k=\sum_{1\leq i\leq m}a_i$.
	
	Then $t$ is a Union-Find tree if and only if the instance $(a_1,\ldots,a_m,k)$ is a positive instance of the $\textsc{Partition}$ problem.
\end{proposition}
\begin{proof}
	(For an example, the reader is referred to Figure~\ref{fig-rank-reduction}.)
	\begin{figure}[h!]
	\centering\begin{tikzpicture}
	[every node/.style={draw,thick,circle,inner sep=1pt}, sibling distance=3pt,scale=0.6]
	\Tree [.$4$ 
	[.$0$ ]
	[.$1$ [.$0$ ]]
	[.$2$ [.$1$ [.$0$ ] ] [.$0$ ]]
	[.$2$ [.$0$ ][.$1$ ] ]	%alma
	[.$2$ [.$0$ ][.$1$ ][.$1$ ]]	%alma
	[.$2$ [.$0$ ][.$1$ ][.$1$ ] [.$1$ ] ]	%alma
	[.$2$ [.$0$ ][.$1$ ][.$1$ ] [.$1$ ] [.$1$ ]]	%alma
	[.$2$ [.$0$ ][.$1$ ] [.$1$ ][.$1$ ] [.$1$ ]]	%alma
	[.$3$ [.$0$ ] [.$1$ [.$0$ ]] [.$0$ ] [.$0$ ] [.$0$ ] [.$0$ ] [.$0$ ]  [.$0$ ] [.$0$ ]]%kosár
	[.$3$ [.$0$ ] [.$1$ [.$0$ ]] [.$0$ ] [.$0$ ] [.$0$ ] [.$0$ ] [.$0$ ]  [.$0$ ] [.$0$ ]] %kosár	
	]
	\end{tikzpicture}
	\caption{The initial flat tree $t$ corresponding to the \textsc{Partition}
	instance $(1,2,3,4,4,k=2)$. The size of each basket is $(1+2+3+4+4)/k=7$.}
	\label{fig-rank-reduction}
\end{figure}	

	Suppose $\mathcal{I}=(a_1,\ldots,a_m,k)$ is a positive instance of the $\textsc{Partition}$ problem.
	Let $H$ stand for the target sum $\frac{\sum a_i}{k}$.
	Let $\mathcal{B}=\{\mathcal{B}_1,\ldots,\mathcal{B}_k\}$ be a solution of $\mathcal{I}$, i.e., $\sum_{i\in\mathcal{B}_j} a_i=H$ for each $j=1,\ldots,k$.
	Let $x_1,\ldots,x_k\in\levelone{t}$ be the nodes corresponding to the baskets of $t$
	and let $y_1,\ldots,y_m\in\levelone{t}$ be the nodes corresponding to the apples of $t$. 
	
	We define the following sequence $t_0,t_1,\ldots,t_m$ of trees: $t_0=t$ and for each $i=1,\ldots,m$, let $t_i=\push(t_{i-1},y_i,x_j)$ with 
	$1\leq j\leq k$ being the unique index with $i\in\mathcal{B}_j$.
	Then, $\levelone{t_m}$ consists of $x_1,\ldots,x_k$ and the three additional nodes having rank $0$, $1$ and $2$.
	Note that the subtrees rooted at the latter three nodes are Union trees.
	Thus, if each of the trees $t_m|_{x_j}$ is a Union-Find tree, then so is $t$.
	
	Consider a subtree $t'=t_m|_{x_j}$.
	By construction, $t'$ is a tree whose root has rank $3$	and has
	\begin{itemize}
		\item $H+1$ children of rank $0$,
		\item a single child of rank $1$, having a child of rank $0$,
		\item and several (say, $\ell$) apple children with total weight $H$.
	\end{itemize}
	We give a method to transform $t'$ into a Union tree.
	First, we push $a_i$ rank-$0$ nodes to each apple child of weight $a_i$.
	After this stage $t'$ has one child of rank $0$, one child of rank $1$ and $\ell$ ``filled'' apple children, having a root of rank $2$, thus the
	root of the transformed $t'$ satisfies the Union condition. We only have to show that each of these ``filled'' apples is a Union-Find tree.
	
	Such a subtree has a root node of rank $2$, $a_i$ depth-one nodes of rank $1$ and $a_i+1$ depth-one nodes of rank $0$. Then,
	one can push into each node of rank $1$ a node of rank $0$ and arrive to a tree with one depth-one node of rank $0$,
	and $a_i$ depth-one nodes of rank $1$, each having a single child of rank $0$, which is indeed a Union tree, showing the claim by Theorem~\ref{thm-uf-push}.
	
	For an illustration of the construction the reader is referred to Figure~\ref{fig-rank-reduction-method}.

	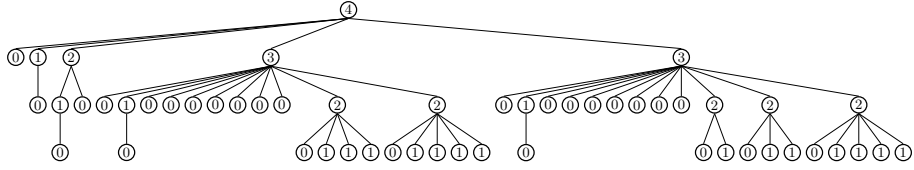
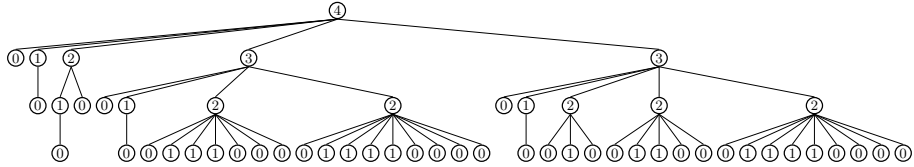
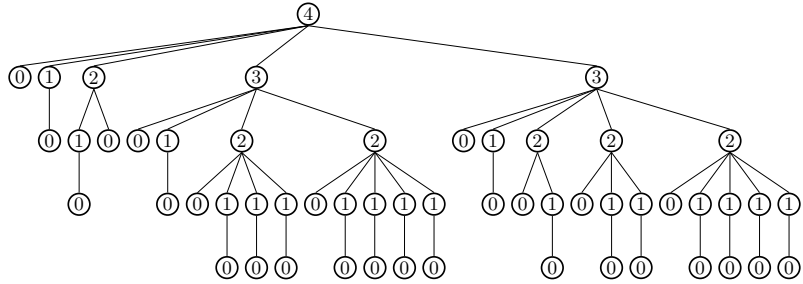
\begin{figure}[h!]
	\begin{subfigure}[c]{\textwidth}
	\centering\begin{tikzpicture}
	[every node/.style={draw,thick,circle,inner sep=1pt}, sibling distance=3pt,scale=0.6]
	\Tree [.$4$ 
	[.$0$ ]
	[.$1$ [.$0$ ]]
	[.$2$ [.$1$ [.$0$ ] ] [.$0$ ]]
	[.$3$ [.$0$ ] [.$1$ [.$0$ ]] [.$0$ ] [.$0$ ] [.$0$ ] [.$0$ ] [.$0$ ]  [.$0$ ] [.$0$ ]
	[.$2$ [.$0$ ][.$1$ ][.$1$ ] [.$1$ ] ]	%alma
	[.$2$ [.$0$ ][.$1$ ] [.$1$ ][.$1$ ] [.$1$ ]]	%alma	
	]%kosár
	[.$3$ [.$0$ ] [.$1$ [.$0$ ]] [.$0$ ] [.$0$ ] [.$0$ ] [.$0$ ] [.$0$ ]  [.$0$ ] [.$0$ ]
	[.$2$ [.$0$ ][.$1$ ] ]	%alma
	[.$2$ [.$0$ ][.$1$ ][.$1$ ]]	%alma
	[.$2$ [.$0$ ][.$1$ ][.$1$ ] [.$1$ ] [.$1$ ]]	%alma
	] %kosár	
	]
	\end{tikzpicture}
	\caption{Apples of size $3$ and $4$ are pushed into the first basket,
		apples of size $1$, $2$ and $4$ are pushed into the second basket.}
	\end{subfigure}
\begin{subfigure}[c]{\textwidth}
	\centering\begin{tikzpicture}
	[every node/.style={draw,thick,circle,inner sep=1pt}, sibling distance=3pt,scale=0.6]
	\Tree [.$4$ 
	[.$0$ ]
	[.$1$ [.$0$ ]]
	[.$2$ [.$1$ [.$0$ ] ] [.$0$ ]]
	[.$3$ [.$0$ ] [.$1$ [.$0$ ]] 
	[.$2$ [.$0$ ][.$1$ ][.$1$ ] [.$1$ ][.$0$ ] [.$0$ ] [.$0$ ] ]	%alma
	[.$2$ [.$0$ ][.$1$ ] [.$1$ ][.$1$ ] [.$1$ ] [.$0$ ] [.$0$ ]  [.$0$ ] [.$0$ ]]	%alma	
	]%kosár
	[.$3$ [.$0$ ] [.$1$ [.$0$ ]]
	[.$2$ [.$0$ ][.$1$ ] [.$0$ ]]	%alma
	[.$2$ [.$0$ ][.$1$ ][.$1$ ] [.$0$ ] [.$0$ ] ]	%alma
	[.$2$ [.$0$ ][.$1$ ][.$1$ ] [.$1$ ] [.$1$ ] [.$0$ ] [.$0$ ]  [.$0$ ] [.$0$ ]]	%alma
	] %kosár	
	]
	\end{tikzpicture}
	\caption{The apples get filled from the baskets' surplus rank-$0$ leaves.}
\end{subfigure}
\begin{subfigure}[c]{\textwidth}
	\centering\begin{tikzpicture}
	[every node/.style={draw,thick,circle,inner sep=1pt}, sibling distance=3pt,scale=0.8]
	\Tree [.$4$ 
	[.$0$ ]
	[.$1$ [.$0$ ]]
	[.$2$ [.$1$ [.$0$ ] ] [.$0$ ]]
	[.$3$ [.$0$ ] [.$1$ [.$0$ ]] 
	[.$2$ [.$0$ ][.$1$ [.$0$ ]][.$1$ [.$0$ ]] [.$1$ [.$0$ ]]]	%alma
	[.$2$ [.$0$ ][.$1$ [.$0$ ]][.$1$ [.$0$ ]][.$1$ [.$0$ ]][.$1$ [.$0$ ]]]	%alma	
	]%kosár
	[.$3$ [.$0$ ] [.$1$ [.$0$ ]]
	[.$2$ [.$0$ ][.$1$ [.$0$ ]]]	%alma
	[.$2$ [.$0$ ][.$1$ [.$0$ ]][.$1$ [.$0$ ]]]	%alma
	[.$2$ [.$0$ ][.$1$ [.$0$ ]][.$1$ [.$0$ ]][.$1$ [.$0$ ]][.$1$ [.$0$ ]]]	%alma
	] %kosár	
	]
	\end{tikzpicture}
	\caption{The filling of the apples is pushed a level deeper and we have a Union tree.}
\end{subfigure}
	\caption{Pushing $t$ of Figure~\ref{fig-rank-reduction} into a Union tree
	according to the solution $3+4=1+2+4$ of the \textsc{Partition} instance.}
		\label{fig-rank-reduction-method}
	\end{figure}
	
	For the other direction, suppose $t$ is a Union-Find tree. By Theorem~\ref{thm-uf-push} and Proposition~\ref{prop-push-is-lift},
	there is a subset $X\subseteq\levelone{t}$ and a mapping $f:Y\to X$ with $Y=\{y_1,\ldots,y_\ell\}=\levelone{t}-X$
	such that for the sequence $t_0=t$, $t_i=\push(t_{i-1},y_i,f(y_i))$ we have that each immediate subtree of $t_\ell$ is
	a Union-Find tree and moreover, the root of $t_\ell$ satisfies the Union condition.
	
	The root of $t$ has rank $4$, $t_\ell$ has to have at least one child having rank $0$, $1$, $2$ and $3$ respectively.
	Since $t$ has exactly one child with rank $0$ and rank $1$, these nodes has to be in $X$.
	This implies that no node gets pushed into the apples at this stage (because the apples have rank $2$).
	Thus, since the apples are \emph{not} Union-Find trees (as they have strictly less rank-$0$ nodes than positive-rank
	nodes, cf. Proposition~\ref{prop-charge}), all the apples have to be in $Y$.
	Apart from the apples, $t$ has exactly one depth-one node of rank $2$ (which happens to be a root of a Union tree),
	thus this node has to stay in $X$ as well. Moreover, we cannot push the baskets as they have the maximal rank $3$,
	hence they cannot be pushed.
	
	Thus, we have to push all the apples, and we can push apples only into baskets (as exactly the baskets
	have rank greater than $2$).
	Let $x\in X$ be a basket node, let $t'$ stand for $t_\ell|_x$ and let $\{y_1',\ldots,y_j'\}\subseteq Y$ be the set
	of those apples that get pushed into $x$ during the operation.
	Then, the total number of nodes having rank $0$ in $t'$ is $H+2+j$
	($j$ of them coming from the apples and the other ones coming from the basket) while the total number
	of nodes having a positive rank is $2+j+A$ where $A$ is the total weight of the apples in $\{y_1',\ldots,y_j'\}$.
	Applying Proposition~\ref{prop-charge} we get that $A\leq H$ for each basket.
	Since the total weight of all apples is $H\cdot k$ and each apple gets pushed into exactly one basket,
	we get that $A=H$ actually holds for each basket.
	Thus, $\mathcal{I}$ is a positive instance of the $\textsc{Partition}$ problem.		
\qed\end{proof}
\begin{theorem}
	The recognition problem of Union-Find trees is $\NP$-complete.
\end{theorem}
\begin{proof}
	By Proposition~\ref{prop-nphard} we get $\NP$-hardness. For membership in $\NP$, 
	we make use of the characterization given in Theorem~\ref{thm-uf-push} and that
	the possible number of pushes is bounded above by $n^2$:
	upon pushing $x$ below $y$, the depth of $x$ and its descendants increases,
	while the depth of the other nodes remains the same.
	Since the depth of any node is at most $n$, the
	sum of the depths of all the nodes is at most $n^2$ in any tree.
	Hence, it suffices to guess nondeterministically a sequence $t=t_0\vdash t_1\vdash \ldots \vdash t_k$
	for some $k\leq n^2$ with $t_k$ being a Union tree (which also can be checked in polynomial time).
\qed\end{proof}
	
\section{Conclusion, future directions}
We have shown that unless $\mathbf{P}=\NP$, there is no efficient algorithm to check whether a given tree is a valid Union-Find tree,
assuming union-by-rank strategy, since the problem is $\NP$-complete, complementing our earlier results
assuming union-by-size strategy. A very natural question is the following: does there exist a merging strategy
under which the time complexity remains amortized almost-constant, and at the same time allows an efficient
recognition algorithm? Although this data structure is called ``primitive'' in the sense that it does not
really need an automatic run-time certifying system, but we find the question to be also interesting
from the mathematical point of view as well.
It would be also an interesting question whether the recognition problem of Union-Find trees built up
according to the union-by-rank strategy is still $\NP$-complete if the nodes of the tree are not tagged
with the rank, that is, given a tree without rank info, does there exist a Union-Find tree with the same
underlying tree?

%\section*{References}
\bibliography{biblio}{}
\bibliographystyle{plain}

\end{document}